\documentclass[conference]{IEEEtran}
\usepackage{amsmath,amsfonts,amsthm}
\usepackage{algorithmic}
\usepackage{algorithm}
\usepackage{array}
\usepackage{textcomp}
\usepackage{stfloats}
\usepackage{url}
\usepackage{color}
\usepackage{verbatim}
\usepackage{graphicx}
\usepackage{cite}
\newtheorem{theorem}{Theorem}
\usepackage{subcaption}
\usepackage{graphicx}
\newcommand{\cstpl}[1]{\left[\rule{0 cm}{#1}\right.}
\newcommand{\cstpr}[1]{\left.\rule{0 cm}{#1}\right]}
\usepackage{geometry}
\begin{document}

\title{Level Crossing Rate Analysis for Optimal Single-user RIS Systems}
\vspace{-5em}
\author{\IEEEauthorblockN{Amy S. Inwood\IEEEauthorrefmark{2}\IEEEauthorrefmark{3}, Peter J. Smith\IEEEauthorrefmark{1}, Philippa A. Martin\IEEEauthorrefmark{2}, Graeme K. Woodward\IEEEauthorrefmark{3}}

    \IEEEauthorblockA{\IEEEauthorrefmark{2}Department of Electrical and Computer Engineering, University of Canterbury, Christchurch, New Zealand}
    \IEEEauthorblockA{\IEEEauthorrefmark{3}Wireless Research Centre, University of Canterbury, Christchurch, New Zealand}
    \IEEEauthorblockA{\IEEEauthorrefmark{1}School of Mathematics and Statistics, Victoria University of Wellington, Wellington, New Zealand}
    email: amy.inwood@pg.canterbury.ac.nz, peter.smith@vuw.ac.nz, (philippa.martin, graeme.woodward)@canterbury.ac.nz}
\maketitle

\begin{abstract}
We analyse the level crossing rate (LCR) of an uplink single-user (SU) reconfigurable intelligent surface (RIS) aided system. It is assumed that the RIS to base station (RIS-BS) channel is deployed as line-of-sight (LoS), and the user (UE)-RIS and UE-BS channels are correlated Rayleigh. For the optimal RIS reflection matrix, we derive a novel and exact analytical LCR expression for when the direct (UE-BS) channel is blocked, i.e. the RIS-only channel. Also, the existing exact expression for the direct-only channel (equivalent to classical maximal-ratio-combining (MRC)) suffers from extreme numerical precision problems when the BS has many elements. Therefore, we propose a new stable and accurate approximation to the LCR of the direct channel. The approximation is based on replacing any small similar eigenvalues of the channel correlation matrix by their average.  We show that increasing the number of elements at the RIS or BS and decreasing channel correlation makes the LCR drop more rapidly for thresholds away from the mean SNR.  Crucially, we find that RIS systems do not significantly amplify temporal variations in the channel. This is particularly beneficial for RIS systems considering the difficulty in acquiring channel state information (CSI).
\end{abstract}


\section{Introduction}
Interest in reconfigurable intelligent surfaces (RIS) continues to grow, driven by their ability to manipulate the wireless channel with very little power consumption. RIS can enhance several aspects of link performance including signal-to-noise-ratio (SNR) and blockage avoidance \cite{wu_towards_2020}. A considerable body of work on RIS now exists, including design, testbeds, and performance analysis \cite{liu_reconfigurable_2021}. Standards development is also underway, with ETSI releasing its first report on RIS \cite{etsi_reconfigurable_2023}.

Typical RIS implementations involve many elements, $N$, in an array, where each element adjusts the phase of the reflected signal to enhance the channel between user (UE) and base-station (BS). Hence, the RIS changes the channel via $N$ phase shifts. Fundamental properties of the channel may also be changed, including the level crossing rate (LCR), which quantifies how often the fading crosses a threshold and is thus used to evaluate temporal changes \cite{ivanis_level_2008, beaulieu_level_2003}.

Work in this area is extremely limited, and \cite{simmons_simulation_2023} is the only LCR study of a RIS-based system. In \cite{simmons_simulation_2023}, a complex system is considered with multiple RIS operating in a cooperative manner. Due to the complexity of this system, the study is necessarily limited to simulation. Hence, we are motivated to evaluate the LCR of a more fundamental system: the single-user RIS system where the RIS to BS link is line-of-sight (LoS) and RIS element phase shifts maximise SNR \cite{nadeem_asymptotic_2020}. For this system, we derive the exact analytical expression for the LCR of the RIS-only channel (when the UE-BS link is blocked, indicated in orange in Fig. \ref{fig:LCR_Paper_System_Diagram}), a novel result. The SNR of the direct-only channel (when the RIS is blocked, indicated in green in Fig. \ref{fig:LCR_Paper_System_Diagram}) collapses to the classical case of maximal-ratio-combining (MRC). Previous work on MRC \cite{ivanis_level_2008} is numerically unstable for large arrays. As a result, another important contribution of this work is providing a novel, accurate and numerically stable approximation to the LCR of MRC in correlated Rayleigh channels with large numbers of antennas. We believe the exact analytical expression of the more complex case, when both RIS and direct links are operational, is intractable. As even accurate approximations are mathematically challenging, this scenario is simulated with analysis left as future work.

Analytical results are verified by simulation, and the effects of varying spatial correlation, system size and channel gain are investigated. All results show that the  RIS channel is similar in nature to the direct channel, and thus RIS operation does not amplify any temporal changes. This is particularly important for RIS as channel estimation is difficult \cite{basharat_reconfigurable_2021}, meaning rapid changes in the channel would be problematic.

\textit{Notation}: Upper and lower boldface letters represent matrices and vectors, respectively. $\mathbf{v}_k$ is the $k$-th element of $\mathbf{v}$, $\mathbf{v}_{i,k}$ is the $k$-th element of $\mathbf{v}_{i}$ and $\mathbf{M}_{r,s}$ is the $(r,s)$-th element of $\mathbf{M}$. $\mathbb{E}[\cdot]$ represents statistical expectation. $\dot{x}=\frac{d}{dt}x(t)$ refers to the first derivative of $x$ with respect to time. $\mathbb{C}$ is the set of complex numbers. $\mathcal{CN}(\boldsymbol\mu,\mathbf{Q})$ represents a complex Gaussian distribution with mean $\boldsymbol\mu$ and covariance matrix $\mathbf{Q}$. $\chi_k^2$ is a chi-squared distribution with $k$ degrees of freedom. $\mathrm{Exp}(\mu)$ is an exponential distribution with a mean of $\mu$. ${}_1F_1(a,b;z)$ is the confluent hypergeometric function and ${}_2F_1(a,b,c;z)$ is the Gaussian hypergeometric function. $\Gamma(\cdot)$ is the complete gamma function. $K(\cdot)$ and $E(\cdot)$ are the complete elliptic integrals of the first and second kind and $J_0(\cdot)$ is the zeroth order Bessel function of the first kind. $(\cdot)^T$ and $(\cdot)^\dagger$ represent the transpose and Hermitian transpose operators. The angle of a vector, $\mathbf{x}$, of length $N$ is denoted $\angle\mathbf{x}=[\angle{x}_1,\dots,\angle{x}_N]^T$ and the exponent $e^{\mathbf{x}}=[e^{{x}_1},\dots,e^{{x}_N}]^T$. $f_X$ is the probability density function (PDF) of a random variable $X$ and $f_{X,Y}$ is the joint PDF of random variables $X$ and $Y$. $\Phi_X$ is the characteristic function (CF) of  $X$ and $\Phi_{X,Y}$ is the joint CF of  $X$ and $Y$. $\mathrm{Var}(X)$ is the variance of $X$.

\section{System Model}\label{sec:sysmodel}
We consider the uplink RIS-aided system in Fig. \ref{fig:LCR_Paper_System_Diagram}, where a RIS panel with $N$ elements is located near a BS with $M$ antennas. One single-antenna user is located near to both. 
\begin{figure}[ht!]
    \centering
    \includegraphics[scale=0.55]{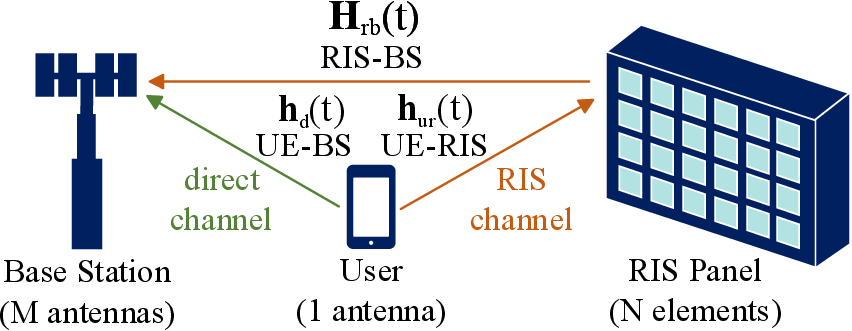}
    \caption{System model showing the uplink channels at time $t$.}
    \label{fig:LCR_Paper_System_Diagram}
\end{figure}

Let $\mathbf{h}_{\mathrm{d}}(t) \in \mathbb{C}^{M\times 1}$, $\mathbf{h}_{\mathrm{ur}}(t) \in \mathbb{C}^{N\times 1}$ be the direct and UE-RIS channels respectively at time instant $t$. $\mathbf{H}_{\mathrm{rb}} \in \mathbb{C}^{M\times N}$ is the static LoS RIS-BS channel. We consider the correlated Rayleigh channels $\mathbf{h}_{\mathrm{d}}(t) = \sqrt{\beta_{\mathrm{d}}}\mathbf{R}_{\mathrm{d}}^{1/2}\mathbf{u}_{\mathrm{d}}$ and $\mathbf{h}_{\mathrm{ur}}(t) = \sqrt{\beta_{\mathrm{ur}}}\mathbf{R}_{\mathrm{ur}}^{1/2}\mathbf{u}_{\mathrm{ur}}$, and the rank-1 LoS channel $\mathbf{H}_{\mathrm{rb}} = \sqrt{\beta_{\mathrm{rb}}}\mathbf{a}_\mathrm{b}\mathbf{a}_\mathrm{r}^\dagger$. $\beta_{\mathrm{d}}$, $\beta_\mathrm{rb}$ and $\beta_{\mathrm{ur}}$ are the channel gains, $\mathbf{a}_\mathrm{b}$ and $\mathbf{a}_\mathrm{r}$ are steering vectors for the LoS ray at the BS and RIS, respectively. $\mathbf{u}_{\mathrm{d}}, \mathbf{u}_{\mathrm{ur}} \sim \mathcal{CN}(\mathbf{0,I})$, and $\mathbf{R}_{\mathrm{d}}$ and $\mathbf{R}_{\mathrm{ur}}$ are spatial correlation matrices for the UE-BS and UE-RIS links, respectively. Correlation is a key feature for RIS systems where a large number of antennas are packed into a comparatively small space. The spatio-temporal correlation structure is assumed to be separable as in \cite{smith_impact_2004}, so that
\begin{align}
    \mathbb{E}\left[\mathbf{h}_{\mathrm{ur},k}(t)\mathbf{h}^*_{\mathrm{ur},l}(t+\tau)\right] & = \beta_\mathrm{ur}\mathbf{R}_{\mathrm{ur},k,l}\rho_\mathrm{ur}(\tau), \\
    \mathbb{E}\left[\mathbf{h}_{\mathrm{d},k}(t)\mathbf{h}^*_{\mathrm{d},l}(t+\tau)\right] & = \beta_\mathrm{d}\mathbf{R}_{\mathrm{d},k,l}\rho_\mathrm{d}(\tau).
\end{align}
In numeric results, we use the classical temporal correlation models $\rho_\mathrm{d}(\tau) = J_0(2\pi f_\mathrm{d}\tau)$ for the direct link and $\rho_\mathrm{ur}(\tau) = J_0(2\pi f_\mathrm{ur}\tau)$ for the RIS-UE link, where $f_\mathrm{d}$ and $f_\mathrm{ur}$ are the Doppler frequencies of the UE with respect to the BS and RIS, respectively, and $\tau$ is the time difference. 

$\mathbf{\Phi}(t)\in \mathbb{C}^{N\times N}$ is a diagonal matrix of RIS element reflection coefficients selected to optimise the total SNR at the BS receiver. The optimal design is given in \cite{singh_optimal_2021} as
\begin{equation}
	\label{eq:Phi}
	\mathbf{\Phi}(t) = \nu_k(t) \, \mathrm{diag}\left(e^{j\angle \mathbf{a}_{\mathrm{r}}}\right)\mathrm{diag}\left(e^{-j\angle \mathbf{h}_{\mathrm{ur}}(t)}\right),
\end{equation}
and $\nu_k(t)\!=\!\mathrm{e}^{j\angle{\mathbf{a}_\mathrm{b}^\dagger\mathbf{h}_{\mathrm{d}}(t)}}\!$. Thus, the received signal at the BS is
\begin{equation}
    \mathbf{r}(t) = (\mathbf{h}_{\mathrm{d}}(t)+\mathbf{H}_{\mathrm{rb}}\mathbf{\Phi}(t) \mathbf{h}_{\mathrm{ur}}(t))s(t) + \mathbf{n}(t), \label{eq:channel}
\end{equation}
and from \cite{singh_optimal_2021}, the SNR is 
\begin{align}
    \label{eq:SNRfull}
    \mathrm{SNR}(t)=&\frac{E_s}{\sigma^2}\Big[\mathbf{h}_{\mathrm{d}}^\dagger(t)\mathbf{h}_{\mathrm{d}}(t)+2\sqrt{\beta_{\mathrm{rb}}}Y(t)|\mathbf{a}_\mathrm{b}^\dagger\mathbf{h}_{\mathrm{d}}(t)| \notag\\ &\qquad +M\beta_\mathrm{rb}Y(t)^2\Big],
\end{align}
where $Y(t) = \sum_{k=1}^{N}|\mathbf{h}_{\mathrm{ur},k}(t)|$, $E_s=\mathbb{E}[|s(t)|^2]$ is the transmitted signal energy and $\sigma^2=\mathbb{E}[|\mathbf{n}(t)|^2]$ is the noise variance.

\section{Analysis}
The level crossing rate is the rate at which the SNR crosses a threshold, $T$, and is defined in \cite{rice_statistical_1948} as
\begin{equation}
    \mathrm{LCR}(T) = \int\displaylimits_{0}^\infty\dot{x}f_{\mathrm{SNR},\mathrm{\dot{SNR}}}(T,\dot{x})d\dot{x}.
\end{equation}
\subsection{LCR of the RIS Channel}\label{sec:RIS}
The SNR for the RIS link (UE - RIS - BS) with no direct path between the TX and RX (setting $\beta_\mathrm{d}=0$  in \eqref{eq:SNRfull}) is
\begin{equation}
    \mathrm{SNR}_\mathrm{R}(t) = \frac{E_s M\beta_\mathrm{rb}}{\sigma^2}Y^2(t) = cY^2(t), \label{eq:RISonlySNR}
\end{equation}
where $c={E_s M\beta_\mathrm{rb}}/{\sigma^2}$. Here, we provide the exact LCR of the RIS (UE-RIS-BS) channel.
\begin{theorem}
    The LCR of the SNR variable, $\mathrm{SNR}_\mathrm{R}(t) =  cY^2(t)$, across a threshold $T$ is given by
    \begin{align}
        \mathrm{LCR_R}(T) &= \sqrt{\frac{2}{\pi}\,c\,T\omega^2} f_{\mathrm{SNR_R}}(T),
    \end{align}
    where $\omega^2$ will be given in (\ref{eq:w2}), $Y(t) = \sum_{k=1}^{N}|\mathbf{h}_{\mathrm{ur},k}(t)|$, $\mathbf{h}_{\mathrm{ur}}(t) \sim \mathcal{CN}(\mathbf{0},\beta_{\mathrm{ur}}\mathbf{R}_{\mathrm{ur}})$ and the elements of $\mathbf{h}_{\mathrm{ur}}(t)$ have temporal correlation  $\mathbb{E}[\mathbf{h}_{\mathrm{ur,k}}(t)\mathbf{h}^*_{\mathrm{ur,k}}(t+\tau)]=J_0(2\pi f_\mathrm{ur}\tau)$.
\end{theorem}
\begin{proof}
    From (\ref{eq:RISonlySNR}), it can be seen that
    \begin{equation}\label{snrdot}
        \dot{\mathrm{SNR}}_\mathrm{R}(t) = 2\,c\,Y(t)\dot{Y}(t) = 2\sqrt{c\,\,\mathrm{SNR}_\mathrm{R}(t)}\dot{Y}(t).
    \end{equation}
    From \cite{beaulieu_level_2003}, $\dot{Y}(t)$ is Gaussian and independent of $Y(t)$. Hence, $\dot{Y}(t)\sim\mathcal{N}(0,\omega^2)$, where, as derived in Appendix B,
    \begin{align}
        \omega^2 &= \sum_{k=1}^N\sum_{l=1}^N\mathbb{E}\left[|\dot{\mathbf{h}}_{\mathrm{ur},k}(t)||\dot{\mathbf{h}}_{\mathrm{ur},l}(t)|\right],\notag\\
        &= \pi^2\!f_\mathrm{ur}^2\beta_\mathrm{ur}\!\sum_{k=1}^N\sum_{l=1}^N\!\left\{\!E(\mathbf{R}_{\mathrm{ur},kl})\!-\!(1\!-\!\mathbf{R}_{\mathrm{ur},kl}^2)K(\mathbf{R}_{\mathrm{ur},kl})\!\right\}\!, \label{eq:w2}
    \end{align}
    and $|\dot{\mathbf{h}(t)}|= \frac{d}{dt}|\mathbf{h}(t)|$. Given $Y(t)$, $\mathrm{\dot{SNR}_R}(t)$ has the conditional distribution $\mathcal{N}(0,4\,c\,\omega^2\mathrm{SNR_R}(t))$. Hence,
    % ASK PETE AROUND DOT OR NOT FOR X 
    \begin{align}
        \mathrm{LCR_R}(T) &= \int\displaylimits_0^\infty \dot{x}\,f_{\mathrm{SNR_R,\dot{SNR}_R}}(T,x)\, d\dot{x}, \notag \\
        & = \int\displaylimits_0^\infty \dot{x}\,f_{\mathrm{SNR_R|\dot{SNR}_R}}(\dot{x}\,|\,T)\,f_{\mathrm{SNR_R}}(T)\, d\dot{x}, \nonumber \\
        & = \frac{f_{\mathrm{SNR_R}}(T)}{\sqrt{8\pi c\,T\omega^2}}\int\displaylimits_0^\infty \dot{x}\exp{\left(\frac{-\dot{x}^2}{8cT\omega^2}\right)}\,d\dot{x}, \nonumber \\ 
        & = \sqrt{\frac{2}{\pi}\,c\,T\omega^2} f_{\mathrm{SNR_R}}(T). \label{eq:LCR_RIS_exact}
    \end{align}\end{proof}

The result in (\ref{eq:LCR_RIS_exact}) has the same general form as the LCR for the SNR of MRC in i.i.d. Rayleigh fading channels \cite{ko_general_2002}. This suggests that 
the RIS does not accentuate temporal variations in the SNR relative to standard MIMO techniques.
An exact expression for $f_\mathrm{SNR_R}(T)$ is believed to be intractable. However, $\mathrm{SNR}_\mathrm{R}(t) = cY^2(t)$ and a gamma approximation to $Y(t)$ has been proposed and validated in \cite{singh_optimal_2021}. Using standard transformation theory to compute the approximate PDF of $\mathrm{SNR}_\mathrm{R}(t)$ from the gamma approximation for $Y(t)$ gives the approximate LCR,
\begin{align}\label{approxRISLCR}
    \mathrm{LCR_R}(T) \approx \frac{1}{\Gamma(r)}\sqrt{\frac{\omega^2}{2\pi}\,}\,\theta^r\left(\frac{T}{c}\right)^{\!\!\frac{r-1}{2}}\!\!\exp\!\left(\!-\theta\sqrt{\frac{T}{C}}\right)\!,
\end{align}
where $\theta=\frac{\mathbb{E}[Y(t)]}{\mathrm{Var}[Y(t)]}$, $r\!=\!\theta\mathbb{E}[Y(t)]$ ensure the  gamma PDF has  the correct mean and variance. These are  $\mathbb{E}[Y(t)]=\frac{N}{2}\sqrt{\pi \beta_\mathrm{ur}}$ and $\mathrm{Var[Y(t)]}\!=\beta_\mathrm{ur}\left(N+F\right)-\frac{N^2 \pi \beta_\mathrm{ur}}{4}$, where $F\!=\!\frac{\pi}{4}\sum^N_{i=1}\!\sum_{j\neq i}{}_2F_1\left(-\frac{1}{2},-\frac{1}{2},1,|\mathbf{R}_{\mathrm{ur},i,j}|^2\right)$  from \cite{singh_optimal_2021}.
\subsection{LCR of the Direct Channel}\label{sec:direct}
In the absence of a RIS, or when the RIS is blocked, the SNR for the direct link (setting $\beta_\mathrm{rb}=0$ in \eqref{eq:SNRfull}), is
\begin{equation}
\label{eq:chandir}
    \mathrm{SNR}_\mathrm{d}(t) = \frac{E_s}{\sigma^2}\mathbf{h}_{\mathrm{d}}^\dagger(t)\mathbf{h}_{\mathrm{d}}(t).
\end{equation}
The LCR of the SNR for this channel is given in \cite{ivanis_level_2008} as
\begin{align}
\label{eq:LCR_direct_exact}
    &\mathrm{LCR}_\mathrm{d}(T) = \frac{\sqrt{\pi}}{3}\!\left({2T}\right)^{\frac{3}{2}}\!\!f_\mathrm{d}\!\sum_{n=1}^M\mathrm{e}^{-\frac{T}{\theta_n}}\sum_{\substack{l=1 \\ l\neq n}}^M\!
    \Bigg(\!\frac{\sqrt{\theta_n}}
    {\theta_l(\theta_n\!-\!\theta_l)} \notag\\ &\quad\times\!{}_1F_1\!\left(\!1,\frac{5}{2};-\frac{T}{\theta_l}\right)\!\!\prod\limits_{m \in S_{l,n}}\!\frac{\theta_l\theta_n}{(\theta_l\!-\!\theta_m)(\theta_n\!-\!\theta_m)}\Bigg),\!\!\!
\end{align}
where $S_{l,n} = \{m\in\mathbb{N}\,|\, 1 \leq m \leq M , m\neq n , m\neq l\}$ and $\theta_1...\theta_M$ are the eigenvalues of $\frac{E_s}{\sigma^2}\beta_\mathrm{d}\mathbf{R}_\mathrm{d}$. 

Note that \cite{ivanis_level_2008} was based on unequal power channels. However, the results can be used here, as $\frac{E_s}{\sigma^2}\mathbf{h}_{\mathrm{d}}^\dagger(t)\mathbf{h}_{\mathrm{d}}(t)=\frac{E_s}{\sigma^2}\beta_{\mathrm{d}} \mathbf{u}_{\mathrm{d}}^\dagger\mathbf{R}_\mathrm{d}\mathbf{u}_{\mathrm{d}}$ is statistically identical to $\mathbf{u}_{\mathrm{d}}^\dagger\mathbf{B}\mathbf{u}_{\mathrm{d}}$, where  $\mathbf{B}={\mathrm{diag}}(\theta_1, \ldots ,\theta_M)$ and $\theta_1, \ldots ,\theta_M$ are the eigenvalues of $\frac{E_s}{\sigma^2}\beta_{\mathrm{d}}\mathbf{R}_\mathrm{d}$. Hence, the SNR with spatial correlation  is statistically identical to the SNR with unequal branch powers.

However, (\ref{eq:LCR_direct_exact}) suffers from extreme numerical precision problems. When spatial correlations are small or large, the eigenvalues of $\mathbf{R}_\mathrm{d}$ become very similar. Products of tiny eigenvalue differences occur in the denominator of (\ref{eq:LCR_direct_exact}), resulting in the addition and subtraction of very large terms. The process loses precision and accurate answers cannot be obtained. This is more likely to occur when $M > 10$.

Since (\ref{eq:LCR_direct_exact}) is unusable for large arrays, we develop a novel approximation which is accurate and numerically stable. The $L$ dominant eigenvalues remain unchanged while the $S$ trailing (similar) eigenvalues are replaced by a single constant equal to their average. Hence, 
$\mathrm{SNR_d} = \sum_{i=1}^{M}\theta_i|\mathbf{u}_{\mathrm{d},i}|^2$ is approximated by $\sum_{i=1}^{L}\lambda_i|\mathbf{u}_{\mathrm{d},i}|^2+\lambda_{L+1}\sum_{i=L+1}^{M}|\mathbf{u}_{\mathrm{d},i}|^2$, where $\lambda_i=\theta_i$ for $i=1,2, \ldots L$, $\lambda_{L+1}=\frac{1}{S}\sum_{i=L+1}^{M}\theta_i$ and $S=M-L$. This is equivalent to replacing $\frac{\beta_{\mathrm{d}} E_s}{\sigma^2}\mathbf{R}_\mathrm{d}$ by $\Lambda=\mathrm{diag}(\lambda_1, \ldots ,\lambda_L, \lambda_{L+1} \ldots ,\lambda_{L+1})$. The LCR for this approximate SNR variable is now derived. 
 
\begin{theorem}\label{Thm1}
    The LCR for the SNR variable, $\mathrm{SNR_a}(t) = \mathbf{g}^\dagger(t)\mathbf{g}(t)$, when $\mathbf{g}(t)\sim\mathcal{CN}(\mathbf{0},\Lambda)$, $\Lambda=\mathrm{diag}(\lambda_1, \ldots ,\lambda_L, \lambda_{L+1} \ldots ,\lambda_{L+1})$, and the elements of the $M \times 1$ vector, $\mathbf{g}(t)$, have temporal correlation  $\mathbb{E}\left[\mathbf{g}_{\mathrm{k}}(t)\mathbf{g}^*_{\mathrm{k}}(t+\tau)\right]=J_0(2\pi f_\mathrm{d}\tau)$, is given by 
\vspace{-0.5em} 
\begin{multline}
    \mathrm{LCR_a}(T) = \frac{\kappa_0j}{4\pi}\sum_{r=1}^L\bigg(B_rI(r,S,r,1) + \sum_{k=1}^SC_kD_{r,k} \\ \times I(r,S-k+1,L+1,k)\bigg), \label{eq:LCRdfinal}
\end{multline} 
where
\begin{align}
    &I(r,p,s,k)\!=\!\sum_{t\neq r}\!\cstpl{0.75cm}\frac{2\pi(-1)^{k+3/2}T^{k+1/2}}{\!\!\prod\limits_{m\neq r,t}\!\!\!(A_{t,r}\!-\!A_{m,r})(A_{t,r}\!-\!A_{L+1,r})^p} \notag \\ &\times\frac{{}_1F_1\!\!\left(k\!+\!\frac{1}{2},k\!+\!\frac{3}{2},(A_{t,r}\!-\!\frac{1}{\lambda_s})T\right)}{\mathrm{e}^{A_{t,r}T}\Gamma\left(k+\frac{3}{2}\right)}\!+ \!\sum_{m=1}^p\frac{2\pi(-1)^{p+k+1/2}}{\Gamma\left(m+k+\frac{1}{2}\right)} \notag \\ &\times\!\frac{T^{m+k-1/2}{}_1F_1\!\!\left(k\!+\!\frac{1}{2},m\!+\!k\!+\!\frac{1}{2},\!(A_{L+1,r}\!-\!\frac{1}{\lambda_s})T)\!\right)}{\mathrm{e}^{A_{L+1,r}T}\!\!\prod\limits_{q\neq r,t}\!\!(A_{t,r}\!-\!A_{q,r})(A_{L+1,r}\!-\!A_{t,r})^{p-m+1}}\!\!\!\cstpr{0.75cm}\!\!,\!\!\! \label{eq:I}
\end{align}
\begin{equation}
    B_r = \frac{\left(\prod_{i=1}^L\lambda_i\right)\lambda_r^{L+S-1/2}\lambda^S_{L+1}\left(2\pi^2f_\mathrm{d}^2\right)^{L+S+1/2}}{\prod_{i\neq r}(\lambda_i-\lambda_r)(\lambda_{L+1}-\lambda_r)^S}, \label{eq:Br}
\end{equation}
\begin{equation}
    C_k = \sum_{s=1}^{k}\frac{(-1)^k}{4^{k-1}}\binom{2s-2}{s-1}\binom{2k-2s}{k-s}\frac{k}{s}, \label{eq:Ck}
\end{equation}
\begin{equation}
    D_{r,k} = \frac{\left(\prod_{i=1}^L\lambda_i\right)\lambda_r^{L+S-k-1}\lambda^{S+3/2}_{L+1}\left(2\pi^2f_\mathrm{d}^2\right)^{L+S+1/2}}{\prod_{i\neq r}(\lambda_i-\lambda_r)(\lambda_{L+1}-\lambda_r)^{S-k+1}}, \label{eq:Drk}
\end{equation}
$\kappa_0 = \left((2\pi^2f_\mathrm{d}^2)^{L+S}\prod_{i=1}^L\lambda_i^2\lambda^{2S}_{L+1}\right)^{-1}$, $A_{i,r} = \frac{\lambda_i+\lambda_r}{\lambda_i\lambda_r}$, $S$ is the number of  repeated values of $\lambda_{L+1}$ in $\Lambda$, $L = M-S$, and $\lambda_1 ... \lambda_L$ are the largest $L$ values in $\Lambda$.
\end{theorem}
\begin{proof}
    See App. A for the derivation of (\ref{eq:LCRdfinal}).
\end{proof}
 %When $k$ is set so the small eigenvalues corresponding to 1\% of the total power are averaged, this approximation performs within x\% of the simulated average. Therefore, it is a very reliable tool to evaluate future large systems with a range of correlations.

\section{Numerical Results}\label{sec:numresults}
Numerical results verify the analysis above and explore the LCR performance of an optimised single-user RIS system for a range of system sizes, link powers and spatial correlations. The effects of UE speed can be seen in the LCR values as they are normalized in the usual way \cite{ivanis_level_2008, beaulieu_level_2003} by Doppler frequency (DF) ($f_m$ represents the relevant DF). The channel gain values were selected according to the distance based path loss model in \cite{wu_intelligent_2019}, where $\beta =  C_0\left({d}/{D_0}\right)^{-\alpha}$,
%\begin{equation}
%    \beta =  C_0\left({d}/{D_0}\right)^{-\alpha},
%\end{equation}
$D_0$ is the reference distance of 1 m, $C_0$ is the pathloss at $D_0$ (-30 dB), $d$ is the link distance in metres and $\alpha$ is the pathloss exponent ($\alpha_{\mathrm{d}} =3.5$, $\alpha_{\mathrm{rb}} = 2$ and $\alpha_{\mathrm{ur}} = 2.8$). Setting $d_\mathrm{rb}$, $d_x$ and $d_y$ in Fig. \ref{fig:power_layout} gives the link distances.

\begin{figure}[ht]
    \centering
    \includegraphics[trim={0cm 0cm 0cm 0cm},clip,scale=0.6]{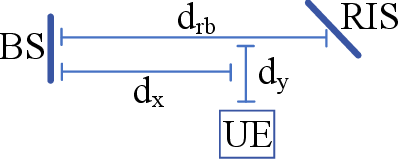}
    \caption{Plan view of the simulation system layout \cite{wu_intelligent_2019}.}
    \label{fig:power_layout}
\end{figure}

Three layouts are used for simulation, all of which include $d_\mathrm{rb} = 40$ m and $d_y = 5$ m. Layout A ($d_x = 29$ m) gives a system where the direct and RIS links have similar power, layout B ($d_x = 20$ m) is a system with a dominant direct link and layout C ($d_x = 35$ m) is  a system with a dominant RIS link. The correlation matrices $\mathbf{R}_\mathrm{d}$ and $\mathbf{R}_\mathrm{ur}$ can represent any  model. For simulation purposes, we use the Rayleigh fading correlation model proposed in \cite{bjornson_rayleigh_2021}, where
\begin{equation}
    \mathbf{R}_{n,m} = \mathrm{sinc} \left( 2 \,d_{mn} \right),\quad n,m = 1, \dots, L,
\end{equation}
$\mathbf{R}\in\{\mathbf{R}_\mathrm{d},\mathbf{R}_\mathrm{ur}\}$, $d_{mn}$ is the Euclidean distance between BS antennas/RIS elements $m$ and $n$, measured in wavelength units, and $L$ is the number of BS antennas/RIS elements.

The steering vectors $\mathbf{a}_\mathrm{b}$ and $\mathbf{a}_\mathrm{r}$ correspond to the vertical uniform rectangular array (VURA) model in  \cite{miller_analytical_2019}. Elements are arranged in a VURA at intervals of $d_b$ and $d_r$ wavelengths at the BS and RIS, respectively. $M_x$ and $N_x$ are the number of BS and RIS elements per row and $M_z$ and $N_z$ are the number of BS and  RIS elements per column, such that $M = M_xM_z$ and $N = N_xN_z$. $\theta_\mathrm{D}$ and $\phi_\mathrm{D}$ are the elevation and azimuth angles of departure (AoDs) at the RIS and $\theta_\mathrm{A}$ and $\phi_\mathrm{A}$ are the corresponding elevation and azimuth angles of arrival (AoAs) at the BS. We assume the RIS is on a $\frac{\pi}{4}$ rad angle with respect to the BS, so $\phi_\mathrm{D} = \frac{5\pi}{4}$ rad and $\phi_\mathrm{A} = \frac{\pi}{4}$ rad. We also assume both are at the same height, so $\theta_\mathrm{D}=\theta_\mathrm{A}=\frac{\pi}{2}$ rad. For all simulations, $10^6$ replicates were generated.

\subsection{Performance Validation and Variation of System Size}
Figs. \ref{fig:systemsizeM} and \ref{fig:systemsizeN} investigate the efficacy of the analytical methods in Theorems 1 and 2, and the impact of element numbers on the LCR for the direct-only (Fig. \ref{fig:systemsizeM}) and RIS-only (Fig. \ref{fig:systemsizeN}) links. Layout A was used for both figures with $d_\mathrm{b} = 0.5$ and $d_\mathrm{r} = 0.1$.
\begin{figure}
    \begin{subfigure}[b]{0.48\textwidth}
        \centering
        \includegraphics[trim={0.9cm 0.13cm 0.9cm 0.11cm},clip,scale=0.4]{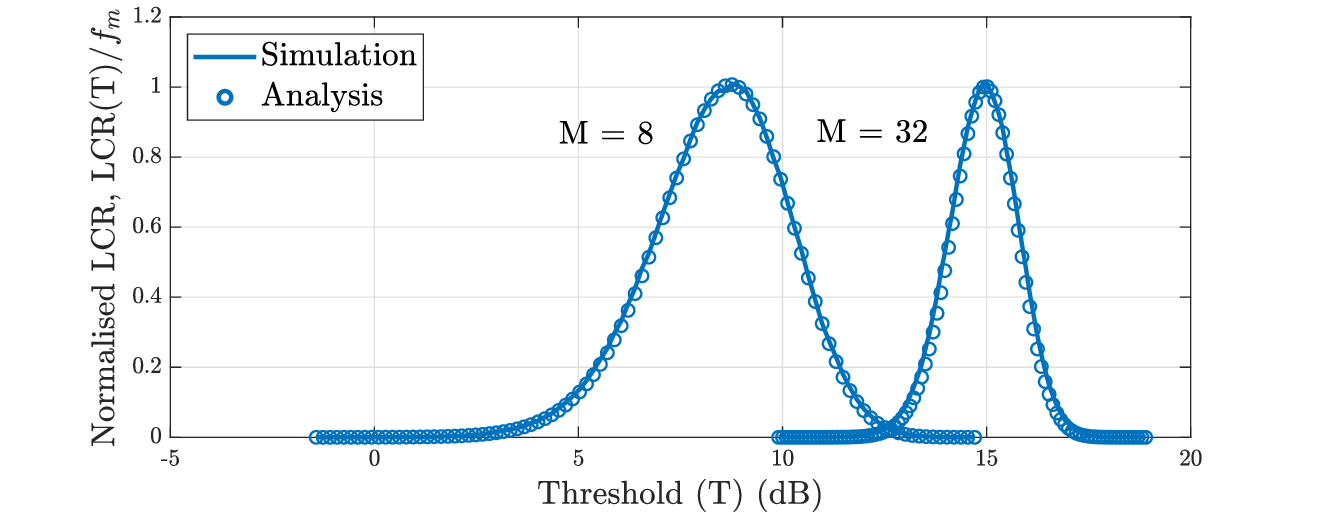}
        \subcaption{Direct link with $M \in \{8, 32\}$ and $M_x \in \{4,8\}$.}
        \label{fig:systemsizeM}
    \end{subfigure}
    \begin{subfigure}[b]{0.48\textwidth}
        \centering
        \includegraphics[trim={0.9cm 0.13cm 0.9cm 0.09cm},clip,scale=0.4]{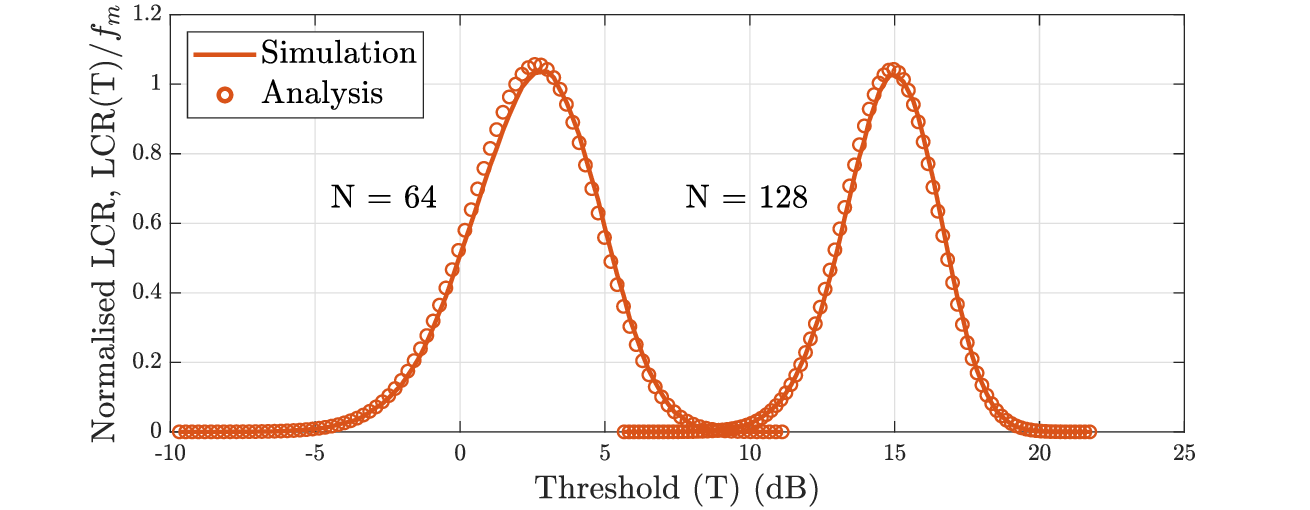}
        \subcaption{RIS link with $N \in \{64, 128\}$ and $N_x \in \{8,16\}$.}
        \label{fig:systemsizeN}
    \end{subfigure}
    \caption{Simulated and analytical LCR comparison.} 
    \vspace{-0.15cm}
\end{figure}
Fig. 3 validates the accuracy of (\ref{eq:LCRdfinal}) and \eqref{approxRISLCR} which are based on Theorems 1 and 2. In Fig. \ref{fig:systemsizeM}, only the 2 leading eigenvalues were kept\footnote{due to being numerically stable and the smallest number where the CDF of the approximation was visibly indistinguishable from the original.}, showing that accuracy is maintained, even while averaging large numbers of  trailing eigenvalues. As expected, the SNR scales up with the number of RIS (N) and BS (M) elements, so the curves slide to the right. The similarity of the LCR curves for RIS and direct links is explored further in Secs. \ref{sec:corr} and \ref{sim}.
\subsection{Variation of Link Powers}
Figs. \ref{fig:power_variationa} and \ref{fig:power_variationb} investigate the impact of the power of each link in a full channel (both direct and RIS links operational) scenario. %, and what occurs when there is a shadowed link (SL). 
Figs. \ref{fig:power_variationa} and \ref{fig:power_variationb} show the LCR for layout B (a dominant direct link) and layout C (a dominant RIS link), respectively. The direct-only, RIS-only and full-channel LCRs are plotted, as well as the full channel LCR with a 50\% power loss in the dominant link (denoted shadowed link (SL)). The array parameters for Fig. 4 are $M = 32$, $M_x = 8$, $N = 128$, $N_x = 16$, $d_\mathrm{b} =0.5$ and $d_\mathrm{r}=0.1$.
\begin{figure}
    \begin{subfigure}[b]{0.48\textwidth}
        \centering
        \includegraphics[trim={0.9cm 0.13cm 0.9cm 0.11cm},clip,scale=0.4]{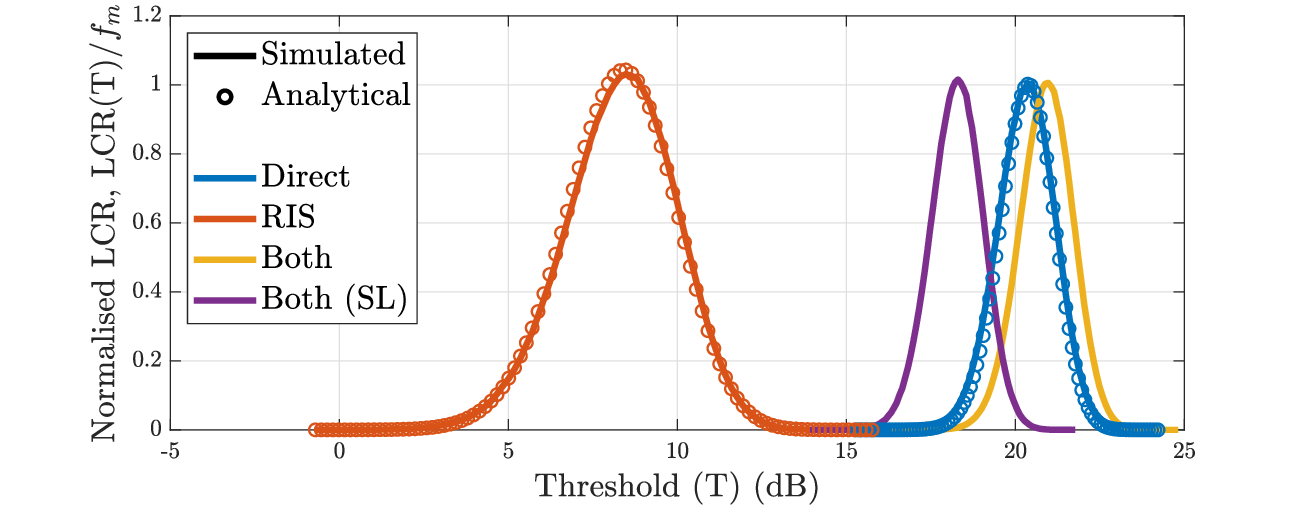}
        \subcaption{Dominant direct link (Layout B).}
        \label{fig:power_variationa}
    \end{subfigure}
    \begin{subfigure}[b]{0.48\textwidth}
        \centering
        \includegraphics[trim={0.9cm 0.13cm 0.9cm 0.09cm},clip,scale=0.4]{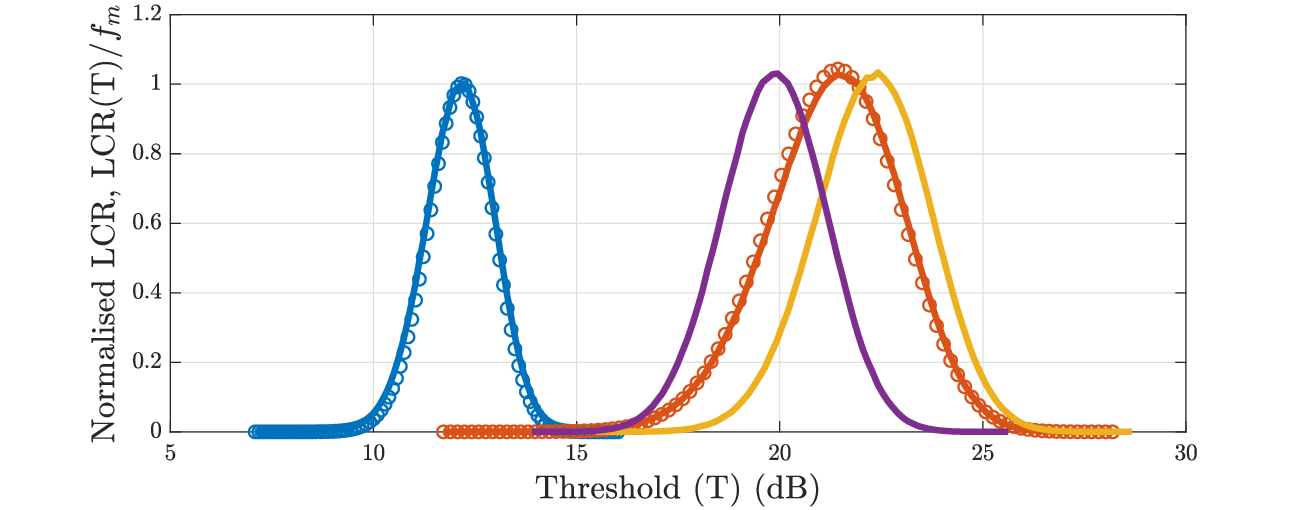}
        \subcaption{Dominant RIS link (Layout C).}
        \label{fig:power_variationb}
    \end{subfigure}
    \caption{Full RIS system shown in Fig. \ref{fig:LCR_Paper_System_Diagram} with and without shadowing applied to the dominant link.}
    \vspace{-0.5cm}
\end{figure}

When the direct and RIS links are combined, the resulting LCR resembles the LCR of the dominant link, with a slight gain in SNR provided by the weaker link shifting the LCR to the right. Reducing the power by 50\% in the dominant link resulted in an approximate 3 dB shift in the LCR curve, but minimal change in shape. This shows that the temporal behaviour is largely determined by the dominant link.

\subsection{Impact of Correlation on LCR}
\label{sec:corr}
Fig. \ref{fig:correlationa} compares LCRs of different correlation levels between elements at both the BS and RIS. Layout A was used for Fig. \ref{fig:correlationa}, with $N = 128$, $N_x = 16$, $M= 32$ and $M_x=8$.  RIS element spacings of $d_\mathrm{r}\in\{0.1,0.5\}$ and BS element spacings of $d_\mathrm{b}\in\{0.5,1\}$ are shown. Layout C was used for Fig. \ref{fig:correlationb}, with $N=M=32$, $N_x=M_x=8$ and $d_\mathrm{r}=d_\mathrm{b}\in\{0.5,1\}$ to investigate the differences between direct and RIS channels under very similar SNR conditions.

\begin{figure}
    \begin{subfigure}[b]{0.48\textwidth}
        \centering
        \includegraphics[trim={0.9cm 0.13cm 0.9cm 0.13cm},clip,scale=0.39]{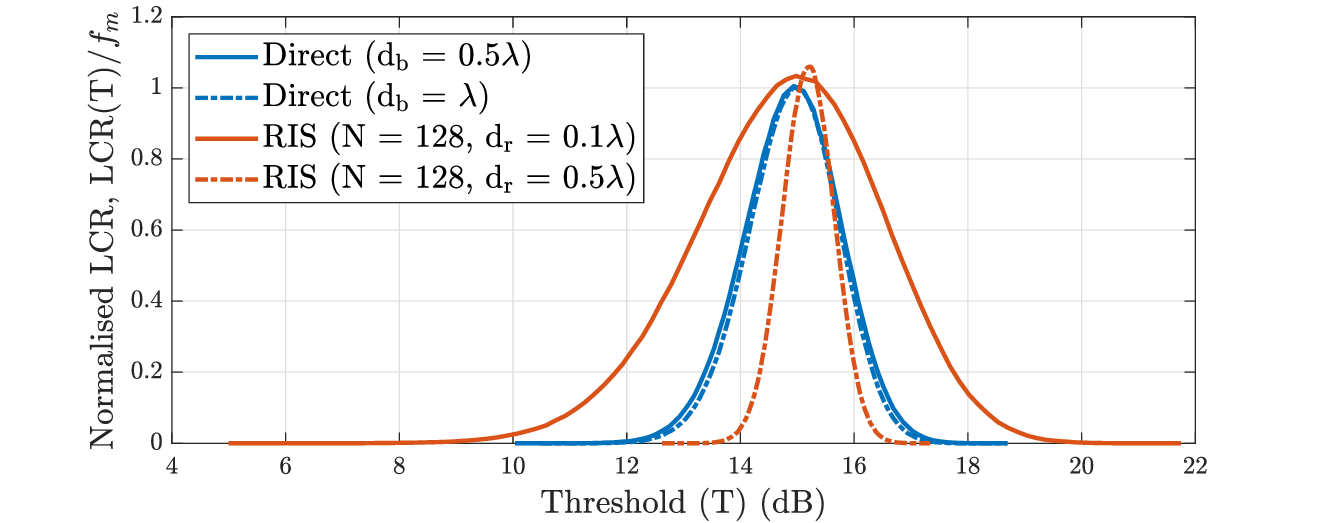}
        \subcaption{Balanced link (Layout A).}
        \label{fig:correlationa}
    \end{subfigure}
    \begin{subfigure}[b]{0.48\textwidth}
    \centering
        \includegraphics[trim={0.9cm 0.13cm 0.9cm 0.13cm},clip,scale=0.39]{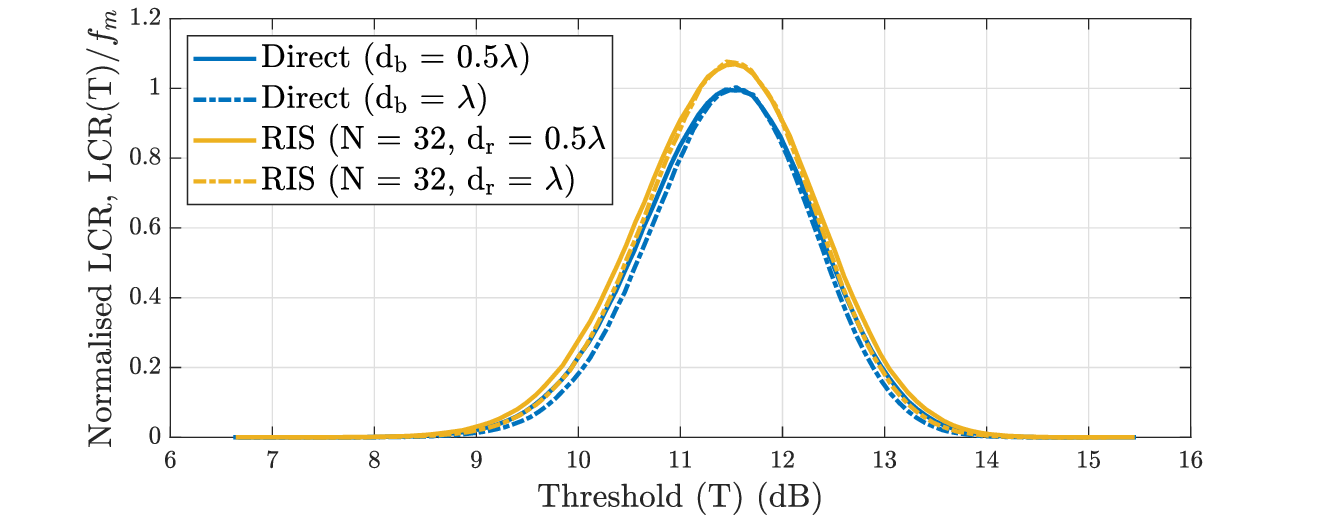}
        \subcaption{Dominant RIS link (Layout C) with $N = 32$.}
        \label{fig:correlationb}
    \end{subfigure}
    \caption{Simulated LCRs for a range of element spacings.}
    \vspace{-0.5cm}
\end{figure}

There is little change in the direct-only LCR when $d_\mathrm{b} = 0.5$ and $d_\mathrm{b} = 1$, as both spacings result in very low correlation. The difference is more noticeable in the RIS-only LCR of Fig. \ref{fig:correlationa} as the smaller spacing of $d_\mathrm{r} = 0.1$ causes noticeably larger correlations. We observe the general trend that the LCR curve narrows as correlation decreases, due to the sums of random variables in (\ref{eq:chandir}) and (\ref{eq:RISonlySNR}). More averaging occurs in these sums when the correlation is low, causing less variance in the SNRs and narrower LCR curves. 

In Fig. \ref{fig:correlationb} we compare direct-only and RIS-only links with equal element numbers and correlations in Layout C, where the SNRs are very similar. The LCRs are very similar, despite the direct-only  SNR in (\ref{eq:chandir}) being a sum of squares and the RIS-only SNR in (\ref{eq:RISonlySNR}) being a sum of amplitudes squared. This similarity is explored mathematically in Sec. \ref{sim}.

\subsection{Special Case: Independent Channels}\label{sim}
To explain the similarity between LCRs for direct-only and RIS-only channels, we consider the case of independent channels.  Using the definitions (\ref{eq:chandir}) and (\ref{snrdot}) and  results in \cite{speed} we define the speed at which the SNR changes by $|\frac{d}{dt}\mathrm{SNR}(t)|$ and obtain the following expressions for the mean speed:
\begin{align}
    \mathbb{E}\left[|\tfrac{d}{dt}\mathrm{{SNR}_d}(t)|\right] &= \pi f_\mathrm{d} \times \frac{\beta_\mathrm{d}E_s}{\sigma^2}\times\frac{\Gamma(M+\frac{1}{2})}{\Gamma(M)}\frac{2\sqrt{2}}{\sqrt{\pi}}, \\
    \mathbb{E}\left[|\tfrac{d}{dt}\mathrm{{SNR}_R}(t)|\right] &= \pi f_\mathrm{ur} \times \frac{\beta_\mathrm{rb}\beta_\mathrm{ur}E_s}{\sigma^2}\times MN^\frac{3}{2}\sqrt{2}.
\end{align}
From (5.11.12, \cite{nist_digital_2023}), $\frac{\Gamma(M+\frac{1}{2})}{\Gamma(M)}\sim \sqrt{M}$ for large $M$. Hence,
\begin{align}
    \mathbb{E}\left[|\tfrac{d}{dt}\mathrm{{SNR}_d}(t)|\right] &\sim \pi f_\mathrm{d} \times \frac{\beta_\mathrm{d}E_s}{\sigma^2}\times\sqrt{\frac{8M}{\pi}}, \label{eq:SNRddot}\\
    \mathbb{E}\left[|\tfrac{d}{dt}\mathrm{{SNR}_R}(t)|\right]\!&=\!\pi f_\mathrm{ur} \!\times\!\frac{M^{\frac{1}{2}}N^{\frac{3}{2}}\beta_\mathrm{rb}\beta_\mathrm{ur}E_s}{\sigma^2}\!\times\! \sqrt{2M}.\label{eq:SNRrdot}
\end{align}
From (\ref{eq:SNRddot}) and (\ref{eq:SNRrdot}), we see that the rates of change contain similar terms for both links: a Doppler term, a channel power term and the final term (approximately $\sqrt{2M}$). The effect of the RIS is seen to be a power scaling, where the RIS path with power $\beta_\mathrm{rb}\beta_\mathrm{ur}$ is scaled by M$^{\frac{1}{2}}N^{\frac{3}{2}}$. Hence, direct and RIS channels will have similar temporal behaviour if their SNR values are similar.

\section{Conclusion}
We have derived LCR expressions for the RIS-only channel and direct-only channel which allow for analytical and numerical comparisons of the two systems. We showed that increasing the number of and spacing between antenna elements at both the BS and RIS reduces LCRs away from the mean SNR. Critically, we showed that the use of a RIS does not significantly amplify temporal changes in the channel, which is a promising result for RIS implementation given the complexity of CSI acquisition. Also, a numerically stable approximation to the LCR of MRC in correlated Rayleigh fading with large numbers of antennas was developed.

\section*{Appendix A \\ Proof of Theorem 1}
Writing $\mathbf{g}=\Lambda^{\frac{1}{2}}\mathbf{u}$ where  $\mathbf{u}\sim \mathcal{CN}(\mathbf{0},\mathbf{I})$, gives
\begin{align}
    \mathrm{SNR_a} = \sum_{i=1}^{L}\lambda_i|\mathbf{u}_{i}|^2 + \lambda_{L+1}\!\!\sum_{i=L+1}^M\!\!|\mathbf{u}_{i}|^2,
\end{align}
where $|\mathbf{u}_{i}|^2\sim\mathrm{Exp}(1)$. Note that $2\sum_{i=L+1}^M|\mathbf{u}_{i}|^2\sim\chi^2_{2S}$ so the problem has changed from a linear combination of exponentials as in \cite{ivanis_level_2008} to a linear combination of exponentials and a single chi-squared random variable. From \cite{ivanis_level_2008},
\begin{align}
    \mathrm{LCR_a}(T) &= \frac{-1}{4\pi^2}\int\displaylimits_{-\infty}^\infty\int\displaylimits_{-\infty}^\infty\!\!\frac{\Phi_{\mathrm{SNR_a},\mathrm{\dot{SNR}_a}}\!(\omega_1,\omega_2)\!-\!\Phi_\mathrm{SNR_a}\!(\omega_1)}{\omega_2^2} \notag \\ &\quad\times \mathrm{e}^{-j\omega_1T}d\omega_1d\omega_2, \notag \\ 
    & = \frac{-1}{4\pi^2}\int\displaylimits_{-\infty}^\infty\!J(\omega_1)\mathrm{e}^{-j\omega_1T}d\omega_1, \label{eq:LCRd}
\end{align}
where $\Phi_\mathrm{SNR_a}$is the CF of $\mathrm{SNR_a}$ and $\Phi_\mathrm{SNR_a,\dot{SNR}_a}$ is the joint CF of $\mathrm{SNR_a}$ and $\mathrm{\dot{SNR}_a}$. Using known results on the CF of exponentials and $\chi^2$ variables, 
\begin{align}
    \Phi_{\mathrm{SNR_a}}(\omega_1) = \prod_{i=1}^{L}\left(\frac{1}{1-j\omega_1\lambda_i}\right)\!\left(\frac{1}{1-j\omega_1\lambda_{L+1}}\right)^S.
\end{align}
To obtain the joint CF, let $r_i=|\mathbf{u}_{i}|$, so that $\mathrm{SNR_a}=\sum_{i=1}^L\lambda_ir_i^2 + \lambda_{L+1}\sum_{i=L+1}^Mr_i^2$ and $\mathrm{\dot{SNR}_a}=2\sum_{i=1}^L\lambda_ir_i\dot{r_i} + 2\lambda_{L+1}\sum_{i=L+1}^Mr_i\dot{r_i}$. It is known that $\dot{r_i}\sim\mathcal{N}(0,\pi^2f_\mathrm{d}^2)$ \cite{ivanis_level_2008} and $\dot{r_i}$ is independent of $r_i$. Hence,
\begin{align}
    \mathrm{\dot{SNR}_a}\sim\mathcal{N}\left(\!0,4\pi^2f_\mathrm{d}^2\bigg(\!\sum_{i=1}^L\!\lambda_i^2r_i^2\!+\!\lambda_{L+1}\!\!\!\!\sum_{i=L+1}^M\!\!\!r_i^2\bigg)\!\!\right),
\end{align}
and the derivative has the representation $\mathrm{\dot{SNR}_a}=\theta Z$, where $\theta^2 = 4\pi^2f_\mathrm{d}^2\left(\sum_{i=1}^L\lambda_i^2r_i^2+\lambda_{L+1}\sum_{i=L+1}^Mr_i^2\right)$, $Z\sim \mathcal{N}(0,1)$ and is independent of $r_i$. Therefore,
\begin{align}
    \!\!\Phi_{\mathrm{SNR_a,\dot{SNR}_a}}\!\!(\omega_1,\omega_2)&\!=\! \mathbb{E}\!\left[\mathrm{e}^{j\omega_1\mathrm{SNR_a} + j\omega_2\mathrm{\dot{SNR}_a}}\right]\notag \\ 
    &\!=\!\mathbb{E}\!\left[\mathbb{E}\!\left[\mathrm{e}^{j\omega_2\theta Z}\right]\mathrm{e}^{j\omega_1\mathrm{SNR_a}}\right] \notag
\end{align}
\begin{multline}
        \Phi_{\mathrm{SNR_a,\dot{SNR}_a}}\!\!(\omega_1,\omega_2)\!=\!\mathbb{E}\Big[\mathrm{exp}\Big(\!\sum_{i=1}^{L}\!\left(j\omega_1\lambda_i\!-2\pi^2f_\mathrm{d}^2\omega_2^2\lambda_i^2\right) \\ \times r_i^2\Big)\mathrm{exp}\left(\!\left(j\omega_1\lambda_{L+1}\!-2\pi^2 f_\mathrm{d}^2\omega_2^2\lambda_{L+1}^2\right)\!W\right)\!\Big],
\end{multline}
where $W=\sum_{i=L+1}^Mr_i^2$. Again, using known results on the CF of exponentials and $\chi^2$ variables,
\begin{align}
    \Phi_{\mathrm{SNR_a,\dot{SNR}_a}}(\omega_1,&\omega_2) = \prod_{i=1}^L\left(\frac{1}{1\!-\!j\omega_1\lambda_i\!+\!2\pi^2f_\mathrm{d}^2\omega_2^2\lambda_i^2}\right) \notag \\ &\!\times\!\left(\frac{1}{1\!-\!j\omega_1\lambda_{L+1}\!+\!2\pi^2f_\mathrm{d}^2\omega_2^2\lambda_{L+1}^2}\right)^{\!\!S}.
\end{align}
Rewriting $J(\omega_1)$ in terms of only the joint CF,
\begin{align}
    \!\!\!J(\omega_1)\!&=\!\!\!\!\int\displaylimits_{-\infty}^\infty\!\!\frac{\Phi_{\mathrm{SNR_a,\dot{SNR}_a}}\!\!(\omega_1,\!\omega_2)\!-\!\Phi_{\mathrm{SNR_a,\dot{SNR}_a}}\!\!(\omega_1,0)}{\omega_2^2}d\omega_2.\!\!\label{eq:Jomega1}
\end{align}
Expanding (\ref{eq:Jomega1}) using (2.102,\cite{gradshteyn_table_1980}) with $\omega_2$ as the variable,
\begin{align}
    \!&J(\omega_1)\!=\!\kappa_0\!\sum_{r=1}^L\!\frac{1}{\prod_{i\neq r}(\!a_r\!-\!a_i\!)}\!\bigg[\frac{1}{(\!a_r\!-\!a_{L+1}\!)^S}\!\!\!\int\displaylimits_{-\infty}^\infty\!\!\frac{1}{(\omega_2^2\!-\!a_r)a_r}   d\omega_2 \notag \\ & \!\!\!+\!\!\sum_{k=1}^S\!\frac{(-1)^{S-k+1}}{(\!a_{L+1}\!-\!a_r\!)^{S-k+1}}\!\sum_{s=1}^k\!\frac{1}{(\!-a_{L+1}\!)^s}\!\binom{k}{s}\!\!\!\int\displaylimits_{-\infty}^\infty\!\!\!\frac{\omega_2^{2(s-1)}}{k(\omega_2^2\!\!-a_{L+1}\!)} d\omega_2\! \bigg]\!, \notag
\end{align}
where $\kappa_0 = \left[(2\pi^2f_\mathrm{d}^2)^{L+S}\prod_{i=1}^L\lambda_i^2\lambda_{L+1}^{2S}\right]^{-1}$, and $a_m = \frac{k\lambda_m\omega_1-1}{2\pi^2f_\mathrm{d}^2\lambda_m^2}$, where $m=1,2,...,L+1$. Applying Cauchy's residue theorem from \cite{jameson_first_1970} and substituting in $a_m$,
\begin{align}
    &J(\omega_1)\!=\!\!\sum_{r=1}^L\!\Bigg[\frac{-\kappa_0\pi jB_r}{\prod\limits_{i\neq r}\!(j\omega_1\!-\!A_{i,r})(j\omega_1\!-\!A_{L+1,r})^{S}(j\omega_1\!\!-\!\frac{1}{\lambda_r})^{3/2}}+ \notag \\
    &\!\sum_{k=1}^S\!\!\frac{-\kappa_0\pi jC_kD_{r,k}}{\prod\limits_{i\neq r}\!(j\omega_1\!\!-\!\!A_{i,r}\!)\!(j\omega_1\!\!-\!\!A_{L+1,r}\!)^{S\!-\!k\!+\!1}\!(j\omega_1\!\!-\!\frac{1}{\lambda_{L\!+\!1}})^{k\!+\!1/2}}\!\Bigg]\!,\!\!\!\!\label{eq:Jw1close}
\end{align}
where $B_r$, $C_k$ and $D_{r,k}$ are given in (\ref{eq:Br}), (\ref{eq:Ck}) and (\ref{eq:Drk}), and $A_{i,r} = \frac{\lambda_i+\lambda_r}{\lambda_i\lambda_r}$. Substituting (\ref{eq:Jw1close}) into (\ref{eq:LCRd}) gives (\ref{eq:LCRdfinal}), where
%\begin{equation}
  %  B_r = \frac{\left(\prod_{i=1}^L\lambda_i\right)\lambda_r^{L+S-1/2}\lambda^S_{L+1}\left(2\pi^2f_d^2\right)^{L+S+1/2}}{\prod_{i\neq r}(\lambda_i-\lambda_r)(\lambda_{L+1}-\lambda_r)^S}, \notag
%\end{equation}
%\begin{equation}
  %  C_k = \sum_{s=1}^{k}\frac{(-1)^k}{4^{k-1}}\binom{2s-2}{s-1}\binom{2k-2s}{k-s}\frac{k}{s}, \notag
%\end{equation}
%\begin{equation}
   % D_{r,k} = \frac{\left(\prod_{i=1}^L\lambda_i\right)\lambda_r^{L+S-k-1}\lambda^{S+3/2}_{L+1}\left(2\pi^2f_d^2\right)^{L+S+1/2}}{\prod_{i\neq r}(\lambda_i-\lambda_r)(\lambda_{L+1}-\lambda_r)^{S-k+1}}, \notag
%\end{equation}
\begin{align}
    I\!(r,\!p,\!s,\!k)\!&\!=\!\!\!\!\!\int\limits_{-\infty}^{\infty}\!\!\frac{\mathrm{e}^{-j\omega_1T}}{\!\prod\limits_{i\neq r}\!\!(j\omega_1\!-\!A_{i,r}\!)\!(j\omega_1\!-\!A_{L\!+\!1,r}\!)^p\!(j\omega_1\!-\!\frac{1}{\lambda_s})^{k\!+\!1\!/\!2}}d\omega_1\!.\! \label{eq:Iint}
\end{align}
Expanding (\ref{eq:Iint}) using partial fractions (2.102,\cite{gradshteyn_table_1980}) and using the integral result in (3.384.7,\cite{gradshteyn_table_1980}) gives (\ref{eq:I}). 
\vspace{-0.6em}
\section*{Appendix B \\ Derivation of $\omega^2$}
\vspace{-0.5em}
Writing $\mathbf{h}_{\mathrm{ur},k}(t+\tau) = Z_k(t+\tau)$ in terms of $\mathbf{h}_{\mathrm{ur},k}(t) = Z_k(t) = r_k\,\mathrm{e}^{j\theta_k}$ and differentiating by first principles gives

\begin{align}\label{deriv}
    \!\!|\dot{Z_k}(t)|\!=\!\lim_{\tau \to 0}\!\frac{|\rho_\mathrm{ur}(\tau) Z_k(t)\!+\!\sqrt{1\!-\!\rho_\mathrm{ur}(\tau)^2}e_k|\!-\!|Z_k(t)|}{\tau},\!\!\!
\end{align}
where $\rho_\mathrm{ur}(\tau) = J_0(2\pi f_\mathrm{ur}\tau)$ and $e_k\sim\mathcal{CN}(0,\beta_\mathrm{ur})$ is independent of $Z_k(t)$. Writing $\rho_\mathrm{ur}(\tau)$ in its series form, and expanding \eqref{deriv} keeping only the leading powers of $\tau$ gives
\begin{align}
    |\dot{Z_k}(t)| = \Re\left(\sqrt{2\pi^2f_\mathrm{ur}^2}\frac{e_k}{Z_k(t)}\right)|Z_k(t)|.
\end{align}
Following the same procedure for $Z_l(t)$ leads to
\begin{align}
    \omega^2 = 2\pi^2f_\mathrm{ur}^2\sum^N_{k=1}\sum^N_{l=1}\mathbb{E}\!\left[\Re\!\left(\!\frac{e_k|Z_k(t)|}{Z_k(t)}\!\right)\!\Re\!\left(\!\frac{e_l|Z_l(t)|}{Z_l(t)}\!\right)\!\right]\!. \label{eq:omega2penultimate}
\end{align}
As $e_k$, $e_l$ are independent of $Z_k(t)$, $Z_l(t)$, we can first take the expectation over these variables. Furthermore, in order to maintain spatial correlation, $\mathbb{E}[e_ke_l^*]=\beta_\mathrm{ur}\mathbf{R}_{\mathrm{ur},kl}$. Hence, 
\begin{align}\label{w2x}
    \omega^2 = \beta_\mathrm{ur}\pi^2f_\mathrm{ur}^2\sum^N_{k=1}\sum^N_{l=1}\mathbf{R}_{\mathrm{ur},kl}\mathbb{E}\left[\mathrm{e}^{j(\theta_k-\theta_l)}\right].
\end{align}
The expectation in \eqref{w2x} is given in (4.32,\cite{miller_complex_1974}), which can be substituted into (\ref{eq:omega2penultimate}) to give (\ref{eq:w2}).\vspace{-0.2em}
\bibliographystyle{IEEEtran}
\bibliography{IEEEabrv, references.bib}

\end{document}